\documentclass[12pt,reqno]{amsart}

\date{\today}

\usepackage[ansinew]{inputenc}
\usepackage{textcomp}
\usepackage{cite}

\usepackage{amsthm}
\usepackage{amssymb}
\usepackage{amsfonts}

\newtheorem{theorem}{Theorem}
\newtheorem{proposition}[theorem]{Proposition}
\newtheorem{lemma}[theorem]{Lemma}
\newtheorem{corollary}[theorem]{Corollary}

\theoremstyle{definition}

\theoremstyle{definition}
\newtheorem*{remark}{Remark}

\newcommand{\C}{\mathbb{C}}
\newcommand{\Cl}{C_{\textnormal{loc}}}

\newcommand{\e}{\textnormal{e}}
\newcommand{\E}{\mathbb{E}}
\newcommand{\G}{\mathcal{G}}

\newcommand{\lk}{\left(}

\newcommand{\N}{\mathbb{N}}
\newcommand{\ol}{\bar {\lambda}}
\newcommand{\p}{\mathbb{P}}

\newcommand{\R}{\mathbb{R}}
\newcommand{\rk}{\right)}
\newcommand{\rso}{random Schr\"odinger operator}

\newcommand{\Z}{\mathbb{Z}}

\DeclareMathOperator{\dist}{dist}

\DeclareMathOperator{\supp}{supp}

\DeclareMathOperator{\argmax}{arg\,max}

\begin{document}

\title[Poisson Statistics on Regular Graphs]{Poisson Eigenvalue Statistics for  Random Schr\"odinger Operators on regular graphs}

\author[Leander Geisinger]{Leander Geisinger}
 \address{Leander Geisinger, Department of Physics, Princeton University, Princeton, NJ 08544, USA}
\email{leander.geisinger@gmail.com}


\begin{abstract}
For random operators it is conjectured that spectral properties of an infinite-volume operator are related to the distribution of spectral gaps of finite-volume approximations. In particular,  localization and pure point spectrum in infinite volume is expected to correspond to Poisson eigenvalue statistics.

Motivated by results about the Anderson model on the infinite tree we consider random Schr\"odinger operators on finite regular graphs. We study local spectral statistics: We analyze the number of eigenvalues in intervals with length comparable to the inverse of the number of vertices of the graph, in the limit where this number tends to infinity. We show that the random point process generated by the rescaled eigenvalues converges  in certain spectral regimes of localization to a Poisson process.

The corresponding result on the lattice was proved by Minami. However, due to the geometric structure of regular graphs the known methods turn out to be difficult to adapt. Therefore we develop a new approach based on direct comparison of eigenvectors.
\end{abstract}

\maketitle


\section{Introduction}

For random operators with extensive disorder it is generally conjectured that there is a connection between spectral properties of a random operator in an infinite-volume set-up and the distribution of spectral gaps of finite-volume approximations, see for example \cite{AltShk86,ShkShaSeaLamSho93,Efe97} and references therein. In particular, it is expected that within  regimes of localization and pure point spectrum of the infinite volume operator the local spectral statistics of finite volume approximations is close to Poisson statistics. 

An intuitive argument that goes back to Dyson \cite{Dys62} can be based on perturbation theory: the eigenvalues of a finite-volume operator repel each other and this level repulsion is comparable to the overlap of the corresponding eigenvectors. For eigenvalues in spectral regimes of localization the eigenvectors are typically localized in different regions and have negligible overlap. Hence, in these spectral regimes level repulsion is expected to have vanishing effect and eigenvalues should be fairly independent. This is the guiding idea, that eigenvalue statistics has to be studied by analyzing eigenvectors.

An important example for random operators is the Anderson model of random Schr\"odinger operators \cite{CarLac90,PasFig92,Sto01}. On the lattice $\Z^d$, $d \geq 1$, a random Schr\"odinger operator  is given by the lattice Laplacian plus a random potential. Anderson localization  on the lattice, in particular the existence of regimes of pure point spectrum was proved  first in \cite{FroSpe83} using a multi-scale analysis and later in \cite{AizMol93,Aiz94} using bounds on fractional moments of the Green function.

On the lattice the aforementioned relation of localization and spectral statistics was also established rigorously: By restricting the random Schr\"odinger operator on the infinite lattice to a cube $C_r$ of side length $r \in \N$ one obtains  a random symmetric $n_r \times n_r$ matrix, where $n_r$ denotes the number of lattice points in $C_r$. So there are $n_r$ eigenvalues $E_j$ and (assuming that the random potential has bounded support) these eigenvalues accumulate in a uniformly bounded interval. To study local spectral statistics one considers the random eigenvalue point process on the scale of the mean eigenvalue spacing, that is on the scale $1/n_r$. Thus one studies the point process generated by the shifted, rescaled points $n_r (E_j - E)$, $j = 1, \dots, n_r$, in the limit $r \to \infty$. If $E$ lies within a spectral regime of localization this random point process converges in distribution to a Poisson process. This was proved in one dimension \cite{Mol80}, where the entire spectrum is pure point, and in higher dimensions \cite{Min96}, where the regime of localization is characterized by exponential decay of fractional moments of the Green function. This result, in particular the relevant regime of localization, was extended in \cite{Wan01,GerKlo13,GerKlo14} and the implications on the distribution of eigenvectors were studied in \cite{Nak06,KilNak07,BelHisSto07}. These results are proved on the lattice $\Z^d$. Poisson eigenvalue statistics was also derived on the single-ended Canopy graph  \cite{AizWar06} but so far no rigorous results have been found for regular graphs.

Here we prove Poisson eigenvalue statistics in certain spectral regimes of localization for a large class of graphs including regular graphs. In a regular graph of degree $K+1$, $K \in \N$, each vertex is connected by an edge to $K+1$ other vertices. As explained below it is difficult to adapt the existing methods to show Poisson statistics to regular graphs so we have to develop a new approach.

The motivation to consider regular graphs stems from the fact that random regular graphs are  appropriate finite-volume approximations of infinite regular trees. A random regular graph of degree $K+1$ with $n$ vertices -- a graph chosen from the ensemble of all such graphs with uniform probability -- typically coincides locally with a regular infinite tree of degree $K+1$. (We refer to \cite{Bol01} for details about random graph models.) This property was used to prove that the spectral measure of the graph Laplacian on a random regular graph converges to the spectral measure on the infinite tree as $n$ tends to infinity  \cite{Kay81}. This convergence was generalized to show that the Laplacian and random Schr\"odinger operators on random regular graphs approximate the corresponding operators on the infinite tree in various ways \cite{BroLin13, DumPal12,TraVuWan13,AnaLem13, Gei13}.

In turn the Anderson model on the infinite tree is one of the most studies models of random operators starting with the seminal work of Anderson \cite{And58}. It is generally conjectured that the Anderson model shows a phase transition from localization to delocalization. However, the existence of delocalized eigenvectors and regimes of absolutely continuous spectrum has been established rigorously only on trees \cite{Kle98,AizSimWar06b,FroHasSpi07,AizWar13}. It is a frequently discussed question if this phase transition can also be seen in the local spectral statistics on regular graphs that approximate trees. It is conjectured that random Schr\"odinger operators on random regular graphs show a transition from localization, where the eigenvalue point process converges to a Poisson process, to delocalization, where eigenvalue statistics is governed by level repulsion familiar from random matrix theory \cite{JakMilRivRud99,AizWar06,DisRiv08,ErdKnoYauYin12,War13}. There are physical arguments and numerical results  \cite{Elo08,OreSmi10,BirReiTar12,DelSca13} but so far there have been no rigorous proofs for neither part of the conjecture.

To explain the difficulties that arise in trying to extend the known methods to prove Poisson statistics to regular graphs, fix a vertex $x$ in a regular graph of degree $K+1$. For $r \in \N$ consider the neighborhood $B_r(x)$ of all vertices that are at distance at most $r$ from $x$. Together with this neighborhood consider also its inner boundary $\partial B_r(x)$ of vertices that are at distance $r$ from $x$. Then both, the number of points in $B_r(x)$ and the number of points in the boundary $\partial B_r(x)$ are of order $K^r$. In particular, the ratio of boundary to volume does not go to zero as $r$ grows. On the lattice $\Z^d$ the same ratio decays with rate $1/r$ and this decay is a crucial ingredient in the existing proofs of Poisson eigenvalue statistics.

For example in \cite{Min96}, eigenvalue statistics is studied by analyzing the Green function in the cube $C_r$. This cube is decomposed into smaller cubes and the Green function is decoupled at the boundaries of the smaller cubes. This decoupling leads to independence and eventually to Poisson statistics. However, this strategy relies on the fact that the error due to decoupling along the boundaries has negligible effect. Hence, this strategy can not be  adapted to regular graphs because of the non-vanishing contribution of boundary terms.  

In this article we study regular graphs  and more general graphs with $n \in \N$ vertices and uniformly bounded degree in the limit $n \to \infty$. This includes all regular graphs with fixed degree and in particular random regular graphs. We show that the rescaled eigenvalue process converges -- in a certain regime of localization that is specified below --   in distribution to a Poisson process. To circumvent the difficulties mentioned above we use a new approach. We do not work primarily with the Green function but we analyze eigenvectors directly. In particular, we do not use an a priori decomposition but we adapt the decomposition to the location of the eigenvectors. This is realized by comparing eigenvectors of  random Schr\"odinger operators on the graph with eigenvectors of local restrictions of the operators.

In the next section we first introduce the relevant notation about graphs, \rso s, and eigenvalue processes. In \eqref{eq:fracmom} we state and discuss the localization assumption given in terms of exponential decay of fractional moments of the Green function. Then in Theorem~\ref{thm:graph} and Corollary~\ref{cor:poisson} we formulate the  main results. In Section~\ref{sec:strategy} we give the strategy of proof and explain the structure of the remainder of the article.


\section{Main result}

We consider simple undirected connected graphs $\mathcal G_n$ with $n \in \N$ vertices.  For two vertices $x, y \in \G_n$ we write $d(x,y)$ for the distance between $x$ and $y$, that is for the length of the shortest path in $\G_n$ connecting $x$ and $y$. For a vertex $x \in \G_n$ and  $r > 0$ let 
\[
B_r(x) = \{ y \in \G_n \, : \, d(x,y) \leq r \}
\]
denote the $r$-neighborhood of $x$. For any subset $\mathcal B \subset \G_n$ we write  $\partial \mathcal B = \{ x \in \mathcal B \, : \, \dist(x,\G_n \setminus \mathcal B ) = 1 \}$ for its inner boundary and $|\mathcal B|$ for the number of vertices in $\mathcal B$. More generally, for a set $A$ we write $|A|$ for the number of elements in $A$, while for an interval $I \subset \R$ we write $|I|$ for its length.

The (maximal) degree of a graph is the maximal number of edges emerging from a vertex. We assume that the degree of the graphs $\mathcal G_n$ is uniformly bounded by $K+1$ with $K \geq 2$. Then for any $n \in \N$, $x \in \G_n$, and $r > 0$ we get
\begin{equation}
\label{volbound}
|B_r(x)| \leq 1+ (K+1) \sum_{m=0}^{r-1} K^m \leq 3 K^r
\end{equation}
 and
 \begin{equation}
 \label{surfbound}
 |\partial B_r(x)| \leq (K+1) K^{r-1} \leq \frac 32 K^r \, .
\end{equation}
This is satisfied, in particular, for $(K+1)$-regular graphs where each vertex is connected by an edge to $K+1$ other vertices. However, our results are valid for all simple undirected connected graphs that satisfy the uniform bounds  \eqref{volbound} and \eqref{surfbound}.

We study  the distribution of eigenvalues of random Schr\"odinger operators
\begin{equation}
\label{eq:op}
H_n(\omega) = A_n + \alpha \, V_n(\omega) \, ,
\end{equation}
with domain $\ell^2(\G_n)$ in the limit $n \to \infty$. Here $A_n$ denotes the adjacency matrix of $\G_n$,
\[
\lk A_n \phi \rk(x) = - \sum_{y \in \G_n \, : \, d(y,x) = 1} \phi(y)  \, , \qquad x \in \G_n \, , \qquad \phi \in \ell^2(\G_n) \, ,
\]
that corresponds to the graph Laplacian with the diagonal terms removed. The random potential  $V_n(\omega)$  acts as a multiplication operator,
\[
(V_n(\omega) \phi)(x) = \omega_x \phi(x) \, ,  \qquad x \in \G_n \, , \qquad \phi \in \ell^2(\G_n) \, ,
\]
where $(\omega_x)_{x \in \G_n}$ is a collection of independent identically distributed real random variables. We assume that the single-site distribution
$\rho(dt) = \textnormal{Prob}(\omega_x \in dt)$, $x \in \G_n$, is absolutely continuous with bounded density such that 
\[
\| \rho \|_\infty = \sup_{t \in \R} |\rho(t)| < \infty \, .
\]
We also assume that the support of $\rho$ is bounded such that $\supp \rho = [-\rho_0 ,\rho _0]$ with $\rho_0 < \infty$. With $\p$ and $\E$ we denote probability  and expectation with respect to the distribution of $\omega = (\omega_x)_{x \in \G_n}$.
Finally $\alpha > 0$ is a parameter controlling the strength of the disorder.

We denote by $(E_j(\omega))_{j=1}^n = \sigma(H_n(\omega))$ the eigenvalues of the operator $H_n(\omega)$ and by $(\phi_j(\omega))_{j=1}^n$ the corresponding orthonormal eigenvectors. For eigenvalues with multiplicities we count them according to the multiplicity and we choose the eigenvectors as an orthonormal basis of the corresponding eigenspace.  (We can assume, however, that almost surely all eigenvalues are simple.) For an interval $J \subset \R$ let $N(J)$ denote the number of eigenvalues in $J$.

The spectrum of $H_n(\omega)$ is almost surely contained in the interval $[ -K-1 - \alpha \rho_0 , K+1+ \alpha \rho_0 ].$ Thus the eigenvalues accumulate in a bounded interval. To study local spectral statistics we are interested in the random eigenvalue point process on the scale of the mean eigenvalue spacing, that is on the scale $1/n$. In particular, we consider random variables $N(I_n)$, where $I_n$ are suitable intervals with length of order $1/n$.
On the infinite $(K+1)$-regular tree the spectrum of the corresponding random Schr\"odinger operator is almost surely contained in the interval $[ -2\sqrt{K} - \alpha \rho_0 , 2\sqrt{K} + \alpha \rho_0 ]$ (see for example \cite{KunSou80,KirMar82,FigNeb91, PasFig92}). On random regular graphs it follows that the support of the spectral measure of $H_n(\omega)$ converges to this interval as $n \to \infty$ \cite{Kay81,Gei13}.

Thus for fixed $E \in   [ -2\sqrt{K} - \alpha \rho_0 , 2\sqrt{K} + \alpha \rho_0 ]$ we consider the random point process $\nu_n^{(E)}$ generated by the rescaled eigenvalues $\left\{ n(E_j - E) \right\}_{j=1}^n$:
\[
\nu_ n^{(E)} = \sum_{j = 1}^n \delta_{n(E_j - E)} \, .
\]
Here $\delta_{n(E_j-E)}$ denotes the Dirac measure: For a Borel set $A \subset \R$,  $\delta_{n(E_j-E)}(A)=1$ if $n(E_j-E) \in A$ and $\delta_{n(E_j-E)}(A)=0$ otherwise.
For any bounded interval $I \subset \R$ we denote by 
\[
I_n = E + I/n = \left\{ t \in \R \, : \, n(t-E) \in I \right\}
\]
the rescaled interval centered at $E$ such that 
\begin{equation}
\label{numberprocess}
\nu_n^{(E)}(I) = N(I_n) \, .
\end{equation}
Our goal is to show that $\nu_n^{(E)}$ converges -- for values of $E$ in suitable spectral regimes of localization -- to a Poisson process.

We prove this convergence to a Poisson process in spectral regimes of localization, where fractional moments of the Green function decay exponentially. A localization assumption in terms of decay of the Green function is also used for example in \cite{Min96,Wan01,AizWar06}. 
To state the precise assumption we denote by 
\[
G_{n,\alpha}(x,y;z) = \left< \delta_x, (H_n(\omega)-z)^{-1} \delta_y \right> \, , \quad x,y \in \G_n \, , \quad z \in \C_+ \, ,
\]
the matrix elements of the Green function of $H_n(\omega)$. Here $\delta_x \in \ell^2(\G_n)$ is such that $\delta_x(u)= 1$ for $x = u$ and $\delta_x(u) = 0$ otherwise.

We assume that there is an open interval $I_0 \subset [ -2 \sqrt{K} - \alpha \rho_0 , 2 \sqrt{K} + \alpha \rho_0 ]$ and that there are constants $n_0 \in \N$ and $s \in (0,1)$ such that for all $n\geq n_0$ and all $z = E + i\zeta$ with $E \in I_0$ and $\zeta > 0$ the estimate
\begin{equation}
\label{eq:fracmom}
\E \left[ \left| G_{n,\alpha}(x,y; z ) \right|^s \right] \leq C_s \exp(-\mu_s \, d(x,y))
\end{equation}
holds for all $x, y \in \G_n$ with $\mu_s$ large enough and with a uniform constant $C_s$. We emphasize that the right-hand side is independent of $n$, $E \in I_0$, and $\zeta > 0$.

Exponential decay of this type can be proved with  established methods developed to derive localization via fractional moments of the Green function  \cite{AizMol93,Aiz94,AizGra98,AizSchFriHun01,AizElgNabSchSto06} (see also \cite[Ch. 7]{AizWar14} for a condensed presentation). In particular, bounds of this form hold for large disorder (that means large $\alpha$) or extreme energies (that means $E$ close to the spectral edges). On the infinite tree even the Green function of the adjacency matrix (corresponding to $\alpha = 0$) decays exponentially with  rate $(\ln K)/2$. Localization and pure point spectrum for \rso s is proved in regimes where \eqref{eq:fracmom} holds with $\mu_s >   \ln K$ so that the bound is summable over the tree \cite{Aiz94}.

However, on trees as well as on finite graphs with bounded degree there are  regimes such that \eqref{eq:fracmom} holds with arbitrarily large decay rate. With the methods mentioned above one can derive, for example, that for given $\mu_0$, there is $\alpha_0$ such that for $\alpha \geq \alpha_0$ the bound \eqref{eq:fracmom} holds with $\mu_s = \mu_0$ for all $E \in [ -2 \sqrt{K} - \alpha \rho_0 , 2 \sqrt{K} + \alpha \rho_0 ]$ and $\zeta > 0$. In this case we could choose $I_0 = ( -2 \sqrt{K} - \alpha \rho_0 , 2 \sqrt{K} + \alpha \rho_0 )$.

\begin{theorem}
\label{thm:graph}
Let $(\G_n)_{n \in \N}$ be a sequence of graphs such that $\G_n$ has $n$ vertices and such that the bounds \eqref{volbound} and \eqref{surfbound} are satisfied for all $n \in \N$.
Assume that $I_0 \subset \R$ is an open interval such that the fractional moment bound \eqref{eq:fracmom} holds with $\mu_s \geq 43 \ln K$. 

Then for all $E \in I_0$ and all bounded intervals $I \subset \R$ the limit 
\[
\lim_{n \to \infty }\left| \E \left[ \exp \lk - t \nu_n^{(E)}(I) \rk \right] - \exp \lk - \E \left[ \nu_n^{(E)}(I) \right]  \lk 1 - e^{-t} \rk \rk \right| = 0
\]
holds uniformly in $t \geq 0$.
\end{theorem}

Let $\mathcal P_\lambda$ denote a random variable, Poisson distributed with parameter $\lambda > 0$. Then
\[
\E \left[ \exp\lk -t \mathcal P_\lambda \rk \right] = \exp \lk - \lambda \lk 1- e^{-t} \rk \rk \, .
\]
Hence,  Theorem~\ref{thm:graph} shows that the Laplace transform of the random variable $\nu_n^{(E)}(I)$ and thus the  distribution of $\nu_n^{(E)}(I)$ is close to Poisson distribution with parameter $\E [ \nu_n^{(E)}(I) ]$.
From general results about convergence of point processes (see for example \cite[Ch. 14]{Kal97}) we deduce that $\nu_n^{(E)}$ converges to a Poisson point process:

Let us wirte $\bar \nu_n^{(E)}$ for the measure $\E [\nu_n^{(E)} ]$. By \eqref{numberprocess}, the Wegner estimate \eqref{wegner} implies
\[
\bar \nu^{(E)}_n(I) \leq \| \rho \|_\infty n |I_n| = \| \rho \|_\infty |I|
\]
for any interval $I \subset \R$. Thus $\bar \nu_n$ is absolutely continuous with respect to Lebesgue measure with uniformly bounded density. In particular,
\[
\sup_{n \in \N} \int_\R f(\tau) \bar \nu_n^{(E)}(d\tau) \leq \| \rho \|_\infty \int_\R f(\tau) d\tau < \infty
\]
for every continuous function $f \, : \, \R \mapsto [0, \infty)$ with compact support. From properties of the vague topology it follows that  there is a subsequence $\lk n_k \rk_{k \in \N}$ and a Borel measure $\bar \nu^{(E)}$  such that $\bar \nu_{n_k}^{(E)}$ converges vaguely to $\bar \nu^{(E)}$  \cite[Thm. A2.3]{Kal97}. In particular, $\bar \nu^{(E)}$ is again absolutely continuous with respect to Lebesgue measure and
\begin{equation}
\label{intmes}
\bar \nu_{n_k}^{(E)}(B)  \to \bar \nu^{(E)}(B) 
\end{equation}
holds for any bounded Borel set $B \subset \R$  \cite[Thm. A2.3]{Kal97}. We combine this with Theorem~\ref{thm:graph} and obtain
\[
\lim_{k \to \infty}  \E \left[ \exp \lk - t \nu_{n_k}^{(E)}(I) \rk \right] = \exp \lk -  \bar \nu^{(E)}(I)   \lk 1 - e^{-t} \rk \rk 
\]
for any bounded interval $I \subset \R$. In particular, $\nu_{n_k}^{(E)}(I)$ converges in distribution to $\mathcal P_{\bar \nu^{(E)}(I)}$. From \cite[Thm. 14.16]{Kal97} we conclude that the process $\nu_{n_k}^{(E)}$ converges in distribution to a Poisson process with intensity measure $\bar \nu^{(E)}$. We summarize these findings in the following statement.

\begin{corollary}
\label{cor:poisson}
Under the conditions of Theorem~\ref{thm:graph}, there is a  subsequence $(\G_{n_k})_{k \in \N}$ and a Borel measure $\bar \nu^{(E)}$ such that the eigenvalue process $\nu_{n_k}^{(E)}$ converges in distribution to a Poisson process with intensity $\bar \nu^{(E)}$. 
\end{corollary}

Let us conclude this section with two remarks about the results.

First we note that the assumption about the rate of decay, $\mu_s \geq 43 \ln K$, is stronger than necessary. It is conjectured that the eigenvalue process converges to a Poisson process at least in regimes where the limiting infinite volume operator has pure point spectrum. On the tree pure point spectrum exists in regimes where \eqref{eq:fracmom} holds with $\mu_s >  \ln K$ \cite{Aiz94}. So our results do not cover the optimal range of values. But they establish existence of spectral regimes where the local distribution of eigenvalues is given by Poisson statistics. So in these regimes they prove the conjectured relation of localization and eigenvalue statistics.

If one is only interested in a statement of the form of Theorem~\ref{thm:graph} without information about the intensity measure, the assumption about the decay could be relaxed to $\mu_s > 18 \ln K$ and maybe even further. However, it seems  to be impossible to reach the optimal condition $\mu_s >  \ln K$ with the methods discussed here. 

The second remark concerns the intensity measure of the limiting Poisson process. On the lattice $\Z^d$ one can choose the subsequence $(n_k)_{k \in \N}$ with $n_k = k^d$ and consider cubes of side length $k$. Then \eqref{intmes} holds for Lebesgue almost every $E$ with 
\[
\bar \nu^{(E)}(B) = D(E) \, |B| \, ,
\]
see for example \cite{Min96}. Here $D(E)$ denotes the density of states at $E$ of the random Schr\"odinger operator on the lattice. Thus the intensity measure is given by Lebesgue measure times the density of states. A similar result holds for the Canopy graph \cite{AizWar06}.  In general, the intensity measure depends on the choice of graphs. But if we consider random regular graphs then it reasonable to conjecture that the limiting intensity measure in Corollary~\ref{cor:poisson} is given by Lebesgue measure times the density of states of the random Schr\"odinger operator on the infinite tree.


\section{Strategy of proof}
\label{sec:strategy}

Our strategy is based on the fact that the bound \eqref{eq:fracmom} on fractional moments of the Green function implies exponential localization of eigenvectors. Indeed, similar as in \cite{Aiz94,AizSchFriHun01} we get the following result about exponential localization that we prove in the appendix.

\begin{proposition}
\label{pro:locas}
For a graph $\G_n$ satisfying \eqref{surfbound} let $I \subset \R$ be such that the fractional moment bound \eqref{eq:fracmom} holds for $E \in I$ with exponent $\mu_s > 2 \ln K$ and let
\[
\ln K <   \mu < \mu_s- \ln K \, .
\]
Then there is a constant $\Cl > 0$ and a  random variable $X_n(\omega)$ with
\begin{equation}
\label{eq:exploc} 
\E \left[ X_n^{2(a-1)/a} \right] \leq  \Cl \, n \, |I| \, , \qquad 1 < a < 2-\frac{\ln K}{\mu} \, ,
\end{equation}
such that the following holds:
For each eigenvalue $E_j \in I$ there is a vertex $x_j \in \G_n$ such that the corresponding $\ell^2(\G_n)$-normalized eigenvector $\phi_j$ satisfies
\begin{equation}
\label{eq:loccond}
|\phi_j(x)| \leq X_n \exp \lk - \mu \, d(x,x_j) \rk 
\end{equation} 
for all $x \in \G_n$.
\end{proposition}

Note that the bound in \eqref{eq:exploc} includes a factor $n$. To get a bound independent of $n$ we will apply Proposition~\ref{pro:locas} to intervals $I_n$ with length of order $1/n$.

The proof of Theorem \ref{thm:graph} is based on comparison of the eigenvalue process $\nu_n^{(E)}$ with an auxiliary process $\eta_n^{(E)}$. In Section \ref{ssec:construction} we construct this process such that 
\[
\eta_n^{(E)}(I) = \sum_{x \in \G_n} b_x^{(E)}(I) \, ,
\] 
where $(b_x^{(E)}(I))_{x \in \G_n}$ are random Bernoulli variables. We define these variables in terms of local operators $H_n^{(x)}(\omega)$. For each $x \in \G_n$ the operator $H_n^{(x)}(\omega)$ is the restriction of $H_n(\omega)$ to the neighborhood $B_{R_n}(x)$, where $R_n \in \N$ is chosen in \eqref{r}. The construction of $\eta_n$ also depends on a small parameter $\tau_n$ chosen in \eqref{tau} that controls how well the local operators $H_n^{(x)}$ approximate $H_n$. 

The definition of $b_x$ in terms of these local operators allows to show that $b_x$ and $b_y$ are independent if $x$ and $y$ are sufficiently far away from each other, see Lemma~\ref{lem:independence}. Based on a local Minami estimate we  also show that $b_x b_y$ is typically zero if $x$ and $y$ are close to each other, see Lemma~\ref{lem:locminami}. In Section \ref{ssec:auxpoisson} we use these facts to prove that $\eta_n$ is close to a Poisson process:

\begin{proposition}
\label{pro:aux}
Let $\G_n$ satisfy \eqref{volbound} and \eqref{surfbound} and let $I \subset \R$ and $E \in \R$. For each $R_n \in \N$ and $0 < \tau_n \leq 1/6 \sqrt K n$ the bound 
\[
\left| \E \left[ \exp \lk - t \eta_n^{(E)}(I) \rk \right] - \exp \lk  - \E \left[ \eta_n^{(E)}(I) \right] \lk 1 - e^{-t} \rk \rk \right| \leq 81 \| \rho \|_\infty^2   \lk |I|+1 \rk^2 K^{8 R_n} n^{-1}
\]
holds for all $t \geq 0$.
\end{proposition}

To complete the proof of Theorem \ref{thm:graph} we have to compare $\nu_n^{(E)}(I)$ and $\mu_n^{(E)}(I)$. The exponential decay of eigenvectors allows to show the following estimate that we prove in Section~\ref{ssec:compare}. 

\begin{proposition}
\label{pro:compare}
Let $\G_n$ satisfy \eqref{volbound} and \eqref{surfbound} and let $I_0 \subset \R$ be such that the fractional moment bound \eqref{eq:fracmom} holds with $\mu_s > 2 \ln K$.  Let $R_n \in \N$ and $\tau_n \leq 1/6 \sqrt K n$ and set
\begin{equation}
\label{c}
C_n =  \tau_n^2 \,  e^{ (\mu_s - 2\ln K)R_n }  \, .
\end{equation}
Then there is a constant $C>0$  such that for each $E \in I_0$ and each  interval $I \subset \R$  the estimates   
\[
 \E \left[ \left| \eta_n^{(E)}(I) -  \nu_n^{(E)}(I) \right| \right]  \leq  C \sqrt K \lk 1 + \| \rho \|_\infty^2 + |I| \rk  \lk \tau_n n+ \tau_n K^{2R_n} + C_n^{-2(a-1)/a} \rk n
\]
and 
\[
 \E \left[ \left| e^{-t \eta_n^{(E)}(I)}  - e^{-t \nu_n^{(E)}(I) }\right| \right]  \leq  C \sqrt K \lk 1 + \| \rho \|_\infty^2 + |I| \rk  \lk \tau_n n+ \tau_n K^{2R_n} + C_n^{-2(a-1)/a} \rk 
\]
hold for all $t \geq 0$ and all $n \in \N$ with $I_n \subset I_0$. Here $1 < a < 2- \ln K / (\mu_s-\ln K)$ denotes the parameter from Proposition~\ref{pro:locas}.
\end{proposition}

Based on these results we can prove Theorem \ref{thm:graph}. We have to choose the parameters $R_n \in \N$ and $\tau_n > 0$  in such way that the error terms in Proposition~\ref{pro:aux} and Proposition~\ref{pro:compare} become small. 

\begin{proof}[Proof of Theorem \ref{thm:graph}]
We estimate
\begin{align*}
\left| \E \left[ e^{  - t \nu_n^{(E)}(I) } \right] - e^{ - \E \left[ \nu_n^{(E)}(I) \right]  \lk 1 - e^{-t} \rk } \right|  \leq \,  &  \E \left[ \left|  e^{ - t \nu_n^{(E)}(I) } -  e^{- t \eta_n^{(E)}(I) } \right| \right]\\
& +  \left|  \E \left[ e^{- t \eta_n^{(E)}(I) } \right] - e^{ - \E \left[ \eta_n^{(E)}(I) \right]  \lk 1 - e^{-t} \rk } \right| \\
& +   1 -    e^{ -  \E \left[ \left| \eta_n^{(E)}(I)  -  \nu_n^{(E)}(I) \right|   \right] \lk 1 - e^{-t} \rk }  \, .
\end{align*}
We remark that $\mu$, $K$, $\|\rho\|_\infty$, and $|I|$ are bounded. We also note that for $E \in I_0$ we have $I_n \subset I_0$ for $n$ large enough. Under the assumption of Theorem~\ref{thm:graph} we can thus apply Proposition~\ref{pro:aux} and Proposition~\ref{pro:compare} to estimate
\begin{equation}
\label{errorterms}
\left| \E \left[ e^{  - t \nu_n^{(E)}(I) } \right] - e^{ - \E \left[ \nu_n^{(E)}(I) \right]  \lk 1 - e^{-t} \rk } \right| \! \leq \! C \! \lk \frac{K^{8R_n}}{n} \! + \! \lk \tau_n n+ \tau_n K^{2R_n} +  C_n^{-2(a-1)/a} \rk n \rk
\end{equation}
with $C_n$ given in \eqref{c} and with a constant $C> 0$ independent of $n$ and $t \geq 0$.

We optimize the error terms
\[
\frac{K^{8 R_n}}n \, , \qquad \tau_n n^2 \, , \qquad \tau_n^{-4(a-1)/a} \, e^{-2 (\mu_s-2\ln K) R_n (a-1)/a} \, n \, .
\]
This leads to
\begin{equation}
\label{tau}
\tau_n = \frac{K^{8R_n}}{n^3}
\end{equation}
and 
\begin{equation}
\label{r}
R_n = \left \lceil  \frac{(7a - 6) \ln n}{(a-1)\mu_s + (18a-14) \ln K } \right \rceil  \, .
\end{equation}
With this choice of parameters all error terms in \eqref{errorterms} are bounded by a constant times
\[
\exp \lk - \frac{(a-1)\mu_s - (38a-34) \ln K} {(a-1) \mu_s +(18a-14) \ln K } \ln n \rk \, .
\]
Thus to ensure that the bound tends to zero we have to show that
\begin{equation}
\label{end}
\lk \frac {\mu_s}{\ln K}  - 34 \rk \frac {a-1}a  > 4 \, .
\end{equation}
The assumption that $\mu_s \geq 43 \ln K$ implies that we can choose $\mu_s > 2 (\sqrt{101} + 11) \ln K$ and thus $a > (4 \sqrt{101} + 41)/(2\sqrt{101} + 21 )$ in  Proposition~\ref{pro:locas}.  This implies \eqref{end} and $\tau_n \leq 1/6 \sqrt K n$ for $n$ large enough and the proof of Theorem~\ref{thm:graph} is complete.
\end{proof}

In the next section we collect some general estimates that will be used in the subsequent proofs. In Section~\ref{sec:poisson} we use the methods from \cite{Che75} to derive an estimate about Poisson approximation of weakly dependent Bernoulli variables. In Section~\ref{sec:aux} we construct the auxiliary process $\eta_n$ and we prove Proposition~\ref{pro:aux} and Proposition~\ref{pro:compare}.


\section{Some general estimates}

In this section we record some general estimates that will be used in the proofs. We consider an arbitrary finite simple undirected connected graph $\G$ with degree $K+1$ and a random Schr\"odinger operator $H(\omega)$  in $\ell^2(\G)$ of the form \eqref{eq:op}. As before we denote by $(E_j)_{j=1}^{|\G|} = \sigma(H)$ and $(\phi_j)_{j=1}^{|\G|}$ the eigenvalues and corresponding $\ell^2(\G)$-normalized eigenvectors of $H(\omega)$. 

First we state the Wegner estimate about the mean number of eigenvalues in an interval \cite{Weg81}. For any $J \subset \R$,
\begin{equation}
\label{wegner}
\E \left[ N(J) \right] \leq \| \rho \|_{\infty} \, | \G |  \, |J| 
\end{equation}
and it follows that 
\begin{equation}
\label{wegner2}
\p \left[ N(J) \geq 1 \right] \leq \| \rho \|_{\infty} \, | \G |  \, |J| \, .
\end{equation}
This fundamental bound was generalized by Minami \cite{Min96} to higher numbers of eigenvalues, see also \cite{BelHisSto07,GraVag07,ComGerKle09,TauVes14}: For $k \in \N$, 
\begin{equation}
\label{minami}
\p \left[ N(J) \geq k \right] \leq \frac 1{k!} \, \| \rho \|_\infty^k \, |\G|^k \, |J|^k \, .
\end{equation}

Next we collect several consequences of exponential localization of eigenvectors.

\begin{lemma}
\label{lem:approx}
Let $\psi \in \ell^2(\G)$ be an approximate eigenvector of $H$. That means we assume that  there is a neighborhood  $\mathcal B \subset \G$ and constants $\lambda \in \R$ and $\tau > 0$  such that $
H \psi(x) = \lambda \psi(x)$ for $ x \in \mathcal B$,
$\psi(x) = 0$ for $x \notin \mathcal B$, $\| \psi \|=1$, and
\[
\sum_{x \in \partial \mathcal B} |\psi(x)|^2 \leq \tau^2 \, .
\]
Then there exists an eigenvalue $E_j \in \sigma(H)$ such that $|\lambda-E_j| \leq \sqrt K \tau$.

Moreover, if $N((\lambda-\epsilon,\lambda+\epsilon)) \leq 1$ for some $\epsilon > 0$ then the  eigenvector $\phi_j$ corresponding to $E_j$ satisfies 
\begin{equation}
\label{eq:evoverlap}
\left| \left< \psi,\phi_j \right> \right|^2 \geq 1 - K \tau^2 \epsilon^{-2}  \, .
\end{equation}
 
\end{lemma}

\begin{proof}
The eigenvectors of $H(\omega)$ form a basis of $\ell^2(\G_n)$, hence
\[
\| (H - \lambda ) \psi \|^2 = \sum_{j = 1}^{|\G|} \left| \left< \psi,\phi_j \right> \right|^2 (\lambda-E_j)^2 \geq \min_j(\lambda-E_j)^2 \, .
\]
Moreover, by assumption,
\[
\| (H - \lambda ) \psi \|^2 = \sum_{x \notin \mathcal B} \left| \sum_{y: d(y,x) = 1} \psi(y) \right|^2 \leq K \sum_{y \in \partial \mathcal B} |\psi(y)|^2 \leq K \tau^2  \, .
\]
Combining these bounds yields the first claim.
To prove the second claim we use these relations again and estimate
\[
\sum_{i:|E_i - \lambda|\geq \epsilon} \left| \left< \psi,\phi_i \right> \right|^2 \leq \frac 1{\epsilon^2} \sum_{i:|E_i-\lambda| \geq \epsilon} \left| \left< \psi,\phi_i \right> \right|^2 (\lambda-E_i)^2 \leq  K \tau^2 \epsilon^{-2} \, .
\]
Since $\psi$ is normalized this  implies 
\[ 
\sum_{i:|E_i-\lambda| < \epsilon} \left| \left< \psi,\phi_i \right> \right|^2 \geq 1 - K  \tau^2 \epsilon^{-2}
\]
and the claim follows from the  assumption $N((\lambda - \epsilon, \lambda+\epsilon)) \leq 1$.
\end{proof}

Let us now show that the assumption about the number of eigenvalues is justified with high probability for appropriate intervals.

\begin{lemma}
\label{lem:double}
Let $I \subset \R$ be a bounded interval and for $\epsilon \leq  |I|$ let $\Delta(\epsilon)$ denote the event that there are two distinct eigenvalues $E_j, E_k \in I \cap \sigma(H)$ satisfying $|E_j - E_k| < \epsilon$. Then
\[
\p \left[ \Delta(\epsilon) \right] \leq 2 \| \rho \|_\infty^2 \, |I| \, \epsilon \, |\G|^2 \, .
\]
\end{lemma}

\begin{proof}
The proof is based on the Minami estimate \eqref{minami}. 
To apply this result we cover the interval $I$ by intervals of length $2 \epsilon$. Choose $a \in \R$ such that $I = [a, a + |I|]$. Let $M(\epsilon)= \lfloor |I|/\epsilon \rfloor$ and for $m = 1, \dots, M(\epsilon)$ we choose intervals $J_m = [a+(m-1) \epsilon, a + (m+1) \epsilon]$.

In the event $\Delta(\epsilon)$ there exists $m \in \{1, \dots, M(\epsilon) \}$ such that $E_j$ and $E_k$ both lie in $J_m$. Hence, by \eqref{minami} with $k=2$, we get
\[
\p \left[\Delta(\epsilon) \right] \leq \p \left[ \exists \, m \in \{1, \dots, M(\epsilon) \} \, : \, N(J_m) \geq 2 \right] \leq M(\epsilon) \frac 12 \|\rho \|^2_\infty |\G|^2 (2\epsilon)^2  
\]
and the result follows from the bound $M(\epsilon) \leq |I|/\epsilon$.
\end{proof}

Finally we need the following simple consequences of the bound \eqref{eq:evoverlap}.

\begin{lemma}
\label{lem:overlap}
Assume there are two normalized vectors $\psi_1, \psi_2 \in \ell^2(\G)$ and an eigenvector $\phi_j$ of $H$ such that the estimates $\left| \left<  \psi_1, \phi_j \right>  \right|^2 \geq 1 - \delta^2$ and $\left| \left< \psi_2, \phi_j \right> \right|^2 \geq 1 - \delta^2$ hold for some $\delta > 0$. Then we have 
\[
\left| \left<  \psi_1, \psi_2 \right> \right| \geq 1- 2\delta^2 \, .
\]
\end{lemma}

\begin{proof}
The eigenvectors of $H$ form a basis of $\ell^2(\G)$, hence 
\[
\left< \psi_1 ,\psi_2 \right> = \sum_{k=1}^{|\G|} \left< \psi_1, \phi_k \right> \left< \phi_k, \psi_2 \right> =  \left< \psi_1, \phi_j \right> \left< \phi_j, \psi_2 \right> + \sum_{k \neq j}  \left< \psi_1, \phi_k \right> \left< \phi_k, \psi_2 \right> \, .
\]
By assumption, $\left| \left< \psi_1,\phi_j \right> \right| \,  \left| \left< \phi_j, \psi_2 \right> \right| \geq 1-\delta^2$ and $\sum_{k \neq j}  \left| \left< \psi_1, \phi_k \right> \right|^2 \leq \delta^2$ and the same bound holds for $\psi_2$. It follows that  
\[
\left| \left< \psi_1, \psi_2 \right> \right| \geq 1- \delta^2 - \lk \sum_{k \neq j}  \left| \left< \psi_1, \phi_k \right> \right|^2 \rk^{1/2} \lk \sum_{k \neq j}  \left| \left< \psi_2, \phi_k \right> \right|^2 \rk^{1/2} \geq 1 - 2 \delta^2 
\]
and the proof is complete.
\end{proof}

\begin{lemma}
\label{lem:overlapbound}
Assume that a normalized vector $\psi \in \ell^2(\G)$ satisfies $\left| \left< \psi, \phi_j \right> \right|^2 \geq 1 - \delta^2$ for an eigenvector $\phi_j$ of $H$ and $\delta > 0$. Then, for each $x \in \G$, 
\[
|\psi(x)| \leq |\phi_j(x)| + \delta \, .
\]
\end{lemma}

\begin{proof}
Again we use that the eigenvectors of $H$ form a basis and estimate
\begin{align*}
|\psi(x)| &\leq \left| \left< \phi_j , \psi \right> \right| \left| \phi_j(x) \right| + \sum_{k \neq j} \left| \left< \phi_k , \psi \right> \right| \left| \phi_k(x) \right| \\
& \leq |\phi_j(x)| + \lk  \sum_{k \neq j} \left| \left< \phi_k , \psi \right> \right| ^2 \rk^{1/2} \lk \sum_{k=1}^{|\G|} |\phi_k(x)|^2 \rk^{1/2} \, .
\end{align*}
Thus the claim follows from the identity $\sum_{k=1}^{|\G|} |\phi_k(x)|^2 = 1$. 
\end{proof}


\section{Poisson approximation}
\label{sec:poisson}

Here we prove the following result about Poisson approximation of a weakly dependent random point  process on a graph $\mathcal G$.  This is the basis for the proof of Proposition~\ref{pro:aux}. We adapt the Chen-Stein method from \cite{Che75}. 

\begin{lemma}
\label{lem:poisson}
Let $(b_x)_{x \in \G}$ be a collection of random Bernoulli variables. Assume that there is $\varrho \in \N$ such that $b_x$ is independent of $(b_y)_{y \notin B_\varrho(x)}$ for all $x \in \G$. We write $\Lambda = \sum_{x \in \G} b_x$ and $\ol= \sum_{x \in \G} \E[b_x]$. Then the bound
\[
\left| \E \left[ e^{-t \Lambda} \right] - e^{  - \ol \lk 1-e^{-t} \rk }\right| \leq  \sum_{x \in \G} \sum_{y \in B_\varrho(x) \setminus \{x\} }  \E \left[ b_x b_y \right] + \sum_{x \in \G} \sum_{y \in B_\varrho(x)}  \E \left[ b_x \right] \E \left[ b_y \right]  
\]
holds for all $t \geq 0$.
\end{lemma}

\begin{proof}
We fix $t \geq 0$ and for $x \in \G$ we write $p_x = \E \left[ b_x \right]$. We define a function $f : \N_0 \to \R$ depending on $t$ and $\ol$. We set $f(0) = 0$ and for $m \in \N$ 
\[
f(m) = - \int_{e^{-t}}^1 e^{- \ol ( y - e^{-t})} y^{m-1} dy  = -(m-1)! \, \ol^{-m}  \sum_{k = 0}^{m-1}   \frac {\ol^k}{k!} \lk  e^{-tk} - e^{- \ol (1-e^{-t})} \rk 
\]
such that, for $m \in \N_0$, 
\begin{equation}
\label{frep}
e^{-t m} - e^{-\ol(1-e^{-t})} = m f(m) - \ol  f(m+1) \, .
\end{equation}
Below we will show that 
\begin{equation}
\label{fdiff}
\left| \E \left[ \Lambda f(\Lambda) - \ol f(\Lambda+1) \right] \right| \leq \sup_{m \in \N_0}  \left| \partial f(m) \right|  \sum_{x \in \G} \! \lk \sum_{y \in B_\varrho(x) \setminus \{x\}} \! \E \left[ b_x b_y  \right] + \sum_{y \in B_\varrho(x)} \! p_x p_y\rk
\end{equation}
with $\partial f (m) = f(m+1) - f(m)$. 
An  elementary estimate shows that $|\partial f(m)| \leq  1$ for all $m \in \N_0$ and all $t \geq 0$ and $\ol \geq 0$. Thus the claim of the proposition follows from \eqref{frep} and \eqref{fdiff}.

To prove \eqref{fdiff} we write 
\[
\Lambda^{(x)} = \sum_{y \in \G \setminus \{x\} } b_y \quad  \textnormal{and}  \quad \Gamma^{(x)} = \sum_{ y \in \G \setminus  B_\varrho(x) } b_y \, .
\]
Then for all $x \in \G$ we have $b_x f(\Lambda) = b_x f(\Lambda^{(x)}+1)$. Hence, we can write 
\begin{align*}
\Lambda f(\Lambda) - \ol f(\Lambda+1) = & \sum_{x \in \G}  b_x \lk f( \Lambda^{(x)} + 1) - f(\Gamma^{(x)} +1 ) \rk + \sum_{x \in \G} (b_x - p_x) f(\Gamma^{(x)} +1 )  \\
& + \sum_{x \in \G} p_x \lk f(\Gamma^{(x)} +1) - f(\Lambda+1) \rk \, .
\end{align*}
By assumption, $b_x$ and $\Gamma^{(x)}$ are independent so the expectation of the second summand is zero. Thus we obtain 
\begin{align}
\nonumber
\E \left[ \Lambda f(\Lambda) - \ol f(\Lambda+1)  \right] = & \sum_{x \in \G} \E \left[  b_x \lk f( \Lambda^{(x)} + 1) - f(\Gamma^{(x)} +1 ) \rk \right] \\
\label{fdiffexp}
& + \sum_{x \in \G} p_x \E \left[  f(\Gamma^{(x)}+1) - f(\Lambda+1) \right]\, .
\end{align}

Let us now fix $x \in \G$ and let $x = y_1, y_2, y_3, \dots, y_{|B_\varrho(x)|}$ denote the vertices in $B_\varrho(x)$. We write $Y_0^{(x)} = \Gamma^{(x)}$ and $Y_j^{(x)} = \Gamma^{(x)} + \sum_{k = 1}^j b_{y_k}$, for $j = 1 , \dots, |B_\varrho(x)|$, such that $Y^{(x)}_{|B_\varrho(x)|} = \Lambda$. Then, for $j = 1, \dots, |B_\varrho(x)|$, we have 
\[
f(Y_{j-1}^{(x)} +1) - f(Y^{(x)}_j + 1) = b_{y_j} \lk f(Y^{(x)}_{j-1} + 1) - f(Y_{j-1}^{(x)} + 2) \rk
\]
and it follows that 
\begin{align}
\nonumber
f(\Gamma^{(x)} +1 ) - f(\Lambda + 1) &= \sum_{j=1}^{|B_\varrho(x)|} \lk f(Y^{(x)}_{j-1} + 1) - f(Y^{(x)}_j +1) \rk  \\
\nonumber
&=   \sum_{j=1}^{|B_\varrho(x)|} b_{y_j} \lk f(Y^{(x)}_{j-1} + 1) - f(Y^{(x)}_{j-1} +2) \rk \\
\label{fdiffy}
&=  - \sum_{j=1}^{|B_\varrho(x)|} b_{y_j} \partial f(Y^{(x)}_{j-1} +1) \, .
\end{align}

Similarly, we write $Z_1^{(x)} = \Gamma^{(x)}$ and $Z_j^{(x)} = \Gamma^{(x)} + \sum_{k = 2}^j b_{y_k}$, for $j = 2, \dots, |B_\varrho(x)|$, such that $Z^{(x)}_{|B_\varrho(x)|} = \Lambda^{(x)}$.  Then, for $j = 2, \dots, |B_\varrho(x)|$, we have 
\[
f(Z_j^{(x)} +1) - f(Z^{(x)}_{j-1} + 1) = b_{y_j} \lk f(Z^{(x)}_{j-1} + 2) - f(Z_{j-1}^{(x)} + 1) \rk
\]
and it follows that 
\begin{align}
\nonumber
f(\Lambda^{(x)} +1 ) - f(\Gamma^{(x)} + 1) &= \sum_{j=2}^{|B_\varrho(x)|} \lk f(Z^{(x)}_j + 1) - f(Z^{(x)}_{j-1} +1) \rk  \\
\nonumber
&=   \sum_{j=2}^{|B_\varrho(x)|} b_{y_j} \lk f(Z^{(x)}_{j-1} + 2) - f(Z^{(x)}_{j-1} +1) \rk \\
\label{fdiffz}
&=   \sum_{j=2}^{|B_\varrho(x)|} b_{y_j} \partial f(Z^{(x)}_{j-1} +1) \, .
\end{align}
Inserting \eqref{fdiffy} and \eqref{fdiffz} into \eqref{fdiffexp} we find
\begin{align*}
\E \left[ \Lambda f(\Lambda) - \ol f(\Lambda+1)  \right] = & \sum_{x \in \G} \sum_{j=2}^{|B_\varrho(x)|} \E \left[ b_x  b_{y_j} \partial f(Z^{(x)}_{j-1}+1) \right] \\
 & - \sum_{x \in \G} \sum_{j=1}^{|B_\varrho(x)|}  p_x \E \left[b_{y_j} \partial f(Y^{(x)}_{j-1}+1) \right]
\end{align*}
and \eqref{fdiff} follows. This completes the proof.
\end{proof}

\section{The auxiliary process}
\label{sec:aux}

In this section we construct an auxiliary process $\eta_n^{(E)}$ and we show that it is close to the eigenvalue process and to a Poisson process. In particular we prove Proposition~\ref{pro:aux} and Proposition~\ref{pro:compare}. We fix a bounded interval $I \subset \R$ and we set  $I_n = E + I/n$.  The construction of $\eta_n$ depends on the parameters $R_n \in \N$ and $0 < \tau_n \leq 1/6\sqrt K n$ specified in \eqref{r} and \eqref{tau}.
To shorten notation we write $B_n(x) = B_{R_n}(x)$.

\subsection{Construction of the process}
\label{ssec:construction}

For each $x \in \G_n$ we define an auxiliary operator $H_n^{(x)}(\omega)$ to be the restriction of $H_n(\omega)$ to $\ell^2 \lk B_n(x) \rk$ with Neumann boundary conditions: 
\[
\lk H_n^{(x)}(\omega) \phi \rk(y) = - \sum_{\begin{subarray}{c} u \in B_n(x) \\ d(u,y) =1 \end{subarray}} \phi(u) + \alpha \, \omega_y \phi(y) 
 \]
for $y \in B_m(x)$ and  $(H_n^{(x)}(\omega) \phi )(y) = 0$ for $y \notin B_m(x)$. We emphasize that the operator $H_n^{(x)}(\omega)$ and therefore also its eigenvectors and eigenvalues are independent of the values of the potential $V_n(\omega)$ outside of the neighborhood $B_n(x)$ and thus depend only on $(\omega_y)_{y \in B_n(x)}$. To shorten notation we do not always write the dependence on $\omega$.

Next we construct a random subset $\mathcal F_n \subset \G_n$ as follows: A vertex $x \in \G_n$ belongs to $\mathcal F_n$ if and only if 
\begin{itemize}
\item[(i)] There exists an eigenvalue $\xi^{(x)} \in \sigma \lk H_n^{(x)} \rk \cap I_n$ and
\item[(ii)] the corresponding  $\ell^2(B_n(x))$-normalized eigenvector $\psi^{(x)}$ satisfies 
\[ 
\sum_{y \in \partial B_n(x)} |\psi^{(x)}(y)|^2 \leq \tau_n^2 \, .
\]
\end{itemize}
For $x \in \mathcal F_n$ we extend the function $\psi^{(x)}$ by zero to $\G_n$ and for $x \in \G_n \setminus \mathcal F_n$ we set $\xi^{(x)} = 0$ and $\psi^{(x)} \equiv 0$ on $\G_n$.

For  $x \in \mathcal F_n$ the function $\psi^{(x)}$ is an approximate eigenvector in the sense of Lemma~\ref{lem:approx}. Indeed, $\| \psi^{(x)} \| = 1$ and $\psi^{(x)}(y) = 0$ for $y \notin B_n(x)$. Moreover, since $H_n^{(x)}$ and $H_n$ coincide locally in $B_n(x)$  this also implies  that $H_n  \psi^{(x)} (y) = \xi^{(x)}  \psi^{(x)}(y)$ for $y \in B_n(x)$. Thus from Lemma~\ref{lem:approx} we get the following error estimate:
For any $x \in \mathcal F_n$ there exists and eigenvalue $E_j$ of $H_n$ such that the bound
\begin{equation}
\label{error}
\left| E_j - \xi^{(x)} \right| \leq \sqrt K \tau_n
\end{equation}
holds.  

\begin{remark}
Loosely speaking, if a vertex $x$ lies in $\mathcal F_n$ then an eigenvector of $H_n$ should be located close to $x$. Indeed, the bound \eqref{eq:evoverlap} from Lemma \ref{lem:approx} shows that typically the eigenvector $\phi_j$ of $H_n$ corresponding to $E_j$ has large overlap with $\psi^{(x)}$. This suggests that $\mathcal F_n$ is decomposed into clusters, neighborhoods around the vertices $x_j$, $j = 1, \dots, \nu_n(I)$, the localization centers of the eigenvectors of $H_n$. So to construct a process close to $\nu_n$ we have to count the number of clusters and we have to thin out the set $\mathcal F_n$.
\end{remark}

To construct a further subset $\mathcal E_n$ of $\mathcal F_n$ we define, for $x \in \mathcal F_n$, a cluster
\[
\mathcal C(x) = \left\{ y \in \mathcal F_n \cap B_{2 R_n}(x) \, : \,  \left| \xi^{(x)} - \xi^{(y)} \right| \leq 2 \sqrt K \tau_n \right\} \, .
\]
For $x \in \G_n$ we set $x \in \mathcal E_n$ if and only if $x \in \mathcal F_n$ and 
\[
\left| \psi^{(x)} (x)\right| \geq \left|\psi^{(y)}(y) \right|
\]
for all $y \in \mathcal C(x)$. We write $\eta_n^{(E)}(I)$ for the number of vertices in $\mathcal E_n$.
Finally, for each $x \in \G_n$ we define a random Bernoulli variable $b_x = 1$ if $x \in \mathcal E_n$  and $b_x = 0$ otherwise so that 
\begin{equation}
\label{eq:auxdef}
\eta_n^{(E)}(I) = |\mathcal E_n| = \sum_{x \in \G_n} b_x \, .
\end{equation}
We remark that the random sets $\mathcal F_n$ and $\mathcal E_n$ and the random variables $(b_x)_{x \in \G_n}$ depend on the interval $I$ and on $E \in I_0$. However, we fix $I$ and $E$ throughout this section, so we often omit writing the dependence on $I$ and $E$.


\subsection{Proof of Propostion \ref{pro:aux}: Independence and Poisson approximation} 
\label{ssec:auxpoisson}

In order to show that $\eta_n$ is close to a Poisson process we use Lemma~\ref{lem:poisson}. To apply this result to the auxiliary process we  need the following lemmas.

\begin{lemma}
\label{lem:independence}
Assume that $x_1, x_2 \in \G_n$ satisfy $d(x_1,x_2) > 6R_n$. Then $b_{x_1}$ and $b_{x_2}$ are independent.
\end{lemma}

\begin{proof}
For any vertex $y \in \G_n$ the event $\{y \in \mathcal F_n \}$ depends only on the eigenvalues and eigenvectors of the operator $H_n^{(y)}$ and is thus measurable with respect to $(\omega_u)_{u \in B_n(y)}$. Accordingly, also the random variables $\xi^{(y)}$ and $|\psi^{(y)}(y)|$ are measurable with respect to $(\omega_u)_{u \in B_n(y)}$.
In turn, any event $\{x \in \mathcal E_n \}$ depends only on the random variables $\xi^{(y)}$ and $|\psi^{(y)}(y)|$ with $y \in B_{2R_n}(x)$. Hence the event $\{ x \in \mathcal E_n \}$ is measurable with respect to 
\[
\bigcup_{y \in B_{2 R_n}(x)} (\omega_u)_{u \in B_n(y)} = (\omega_u)_{u \in B_{3R_n}(x)} \, .
\]
We see that the Bernoulli variable $b_x$ depends only on the random potential in the neighborhood $B_{3R_n}(x)$. This implies the result.
\end{proof}

\begin{lemma}
\label{lem:locminami}
For two distinct vertices $x, y \in \G_n$ we have
\[
\E \left[ b_x b_y \right] \leq  18 \| \rho \|_\infty^2 \lk |I| + 6 \sqrt K n \tau_n \rk^2  K^{2 R_n} n^{-2} \, .
\]
\end{lemma}

\begin{proof}
We consider the subset $B_n(x,y) = B_n(x) \cup B_n(y)$ and we define the operator $H_n^{(x,y)}$ to be the restriction of $H_n$  to $\ell^2(B_n(x,y))$ with Neumann boundary conditions.

Let $\tilde I_n \subset  \R$ denote the interval constructed from $I_n$ by enlarging it by $3 \sqrt K \tau_n$ at both ends.
Assume that $b_x = b_y = 1$. Then we will show that the operator $H_n^{(x,y)}$ has almost surely more than one eigenvalue in $\tilde I_n$.

We argue by contradiction and assume that there is only one eigenvalue in $\tilde I_n$.
The assumption $b_x = b_y = 1$ implies that both $x$ and $y$ lie in $\mathcal F_n$. By definition of $\mathcal F_n$ there are eigenvalues $\xi^{(x)} \in \sigma(H_n^{(x)}) \cap I_n$ and  $\xi^{(y)} \in \sigma(H_n^{(y)}) \cap I_n$ and the corresponding eigenvectors satisfy 
\[
\sum_{u \in \partial B_n(x)} |\psi^{(x)}(u)|^2 \leq \tau_n^2 \qquad \textnormal{and} \qquad \sum_{u \in \partial B_n(y)} |\psi^{(y)}(u)|^2 \leq \tau_n^2 \, .
\]
We extend $\psi^{(x)}$ and $\psi^{(y)}$ by zero to $B_n(x,y)$ and use them as approximate eigenvectors for the operator $H_n^{(x,y)}$.
Applying Lemma~\ref{lem:approx} with $\G = B_n(x,y)$, $H = H_n^{(x,y)}$, and $\lambda = \xi^{(x)}$ and $\lambda = \xi^{(y)}$ respectively yields eigenvalues $\mu^{(x)}$ and $\mu^{(y)}$ of $H_n^{(x,y)}$ satisfying $|\mu^{(x)} - \xi^{(x)}| \leq \sqrt K \tau_n$ and  $|\mu^{(y)} - \xi^{(y)}| \leq \sqrt K \tau_n$. Since both $\xi^{(x)}$ and $\xi^{(y)}$ lie in $I_n$ we see that both $\mu^{(x)}$ and $\mu^{(y)}$ lie in $\tilde I_n$. By assumption, it follows that $\mu^{(x)} = \mu^{(y)}$ and that
\begin{equation}
\label{eq:evnear}
\left| \xi^{(x)} - \xi^{(y)} \right| \leq 2 \sqrt K \tau_n \, .
\end{equation}

Moreover, the assumption that $\tilde I_n$ contains only one eigenvalue of $H_n^{(x,y)}$ also implies that there is only one eigenvalue  in $(\xi^{(x)} - 2 \sqrt K \tau_n, \xi^{(x)} + 2 \sqrt K \tau_n)$ and that there is only one eigenvalue in $(\xi^{(y)} - 2 \sqrt K \tau_n, \xi^{(y)} + 2 \sqrt K \tau_n)$. Thus Lemma~\ref{lem:approx} also gives existence of   an eigenvector $\Phi$ of $H_n^{(x,y)}$ corresponding to $\mu^{(x)} = \mu^{(y)}$ such that 
\[
\left| \left<  \Phi, \psi^{(x)} \right> \right| \geq 1 -  K \tau_n^2 (2 \sqrt K \tau_n)^{-2} = \frac 34 \qquad \textnormal{and} \qquad \left| \left<  \Phi, \psi^{(y)} \right>  \right| \geq \frac 34 \, .
\]
Thus Lemma \ref{lem:overlap} yields $\left| \left<  \psi^{(y)}, \psi^{(x)} \right> \right| > 0$ and in particular 
\begin{equation}
\label{eq:posnear}
d(x,y) \leq 2R_n \, .
\end{equation}
The bounds \eqref{eq:evnear} and \eqref{eq:posnear} show that $y \in \mathcal C(x)$ and $b_x = b_y = 1$ implies $|\psi^{(x)}(x)| = |\psi^{(y)}(y)|$ which is false almost surely for $x \neq y$. We have reached a contradiction.

We have shown that $b_x b_y = 1$ implies that the operator $H_n^{(x,y)}$ has almost surely more than one eigenvalue in the interval $\tilde I_n$. Hence, the Minami estimate \eqref{minami} gives
\[
\E\left[ b_x b_y \right] = \p \left[ b_x b_y = 1 \right] \leq \p \left[ \left| \sigma(H_n^{(x,y)}) \cap \tilde I_n \right| \geq 2\right] \leq \frac 12 \| \rho \|_\infty^2 |\tilde I_n|^2 |B_n(x,y)|^2 \, .
\]
Inserting $| \tilde I_n| = (|I| + 6 \sqrt K \tau_n n)/n$ and $|B_n(x,y)| \leq 2 |B_n(x)| \leq 6 K^{R_n}$ completes the proof.
\end{proof}

\begin{proof}[Proof of Proposition \ref{pro:aux}]
In view of Lemma \ref{lem:independence} we can apply Lemma~\ref{lem:poisson} with $\G = \G_n$ and $\varrho =6 R_n$.   It remains to estimate $\E \left[ b_x \right]$ and $\E \left[ b_x b_y \right]$ for $x \neq y$.  

From Lemma~\ref{lem:locminami}, \eqref{volbound}, and $\tau_n \leq 1/6\sqrt K n$ we obtain, for any $x \in \G_n$, 
\begin{align*}
\sum_{y \in B_{6R_n}(x) \setminus \{ x \} }  \E \left[ b_x b_y \right]  & \leq 18 \| \rho \|_\infty^2 \lk |I| + 6 \sqrt K n \tau_n \rk^2 K^{2R_n}  n^{-2} |B_{6R_n}(x)| \\
& \leq  54 \| \rho \|_\infty^2 \lk |I| + 1 \rk^2 K^{8R_n}  n^{-2} \, .
\end{align*}
It follows that  
\begin{equation}
\label{expbbound1}
\sum_{x \in \G_n} \sum_{y \in B_{6R_n}(x) \setminus \{ x \} }   \E \left[ b_x b_y \right] \leq  54 \| \rho \|_\infty^2 \lk |I| + 1 \rk^2  K^{8R_n}  n^{-1} \, .
\end{equation}

To estimate $\E \left[ b_x \right]$ we note that $b_x = 1$ implies, in particular, $x \in \mathcal F_n$. So applying the Wegner estimate \eqref{wegner2} we obtain, for all $x \in \G_n$,
\[
\E \left[ b_x \right] = \p \left[b_x = 1 \right] \leq \p \left[ \left| \sigma \lk H_n^{(x)} \rk \cap I_n \right| \geq 1 \right] \leq \| \rho \|_\infty |I_n| |B_n(x)| \leq 3 \| \rho \|_\infty |I| K^{R_n} n^{-1} \, .
\]
Hence, it follows that 
\begin{equation}
\label{expbbound2}
\sum_{x \in \G_n} \sum_{y \in B_{6 R_n}} \E \left[ b_x \right] \E \left[ b_y \right] \leq  27 \| \rho \|_\infty^2 |I|^2 K^{8 R_n} n^{-1} \, .
\end{equation}
Thus the claim follows from Lemma~\ref{lem:poisson} and \eqref{expbbound1} and \eqref{expbbound2}.
\end{proof}


\subsection{Proof of Proposition \ref{pro:compare}: Error estimates}
\label{ssec:compare} 

In this subsection we show that  $\eta_n^{(E)}$ is close to the eigenvalue process $\nu_n^{(E)}$. 
We consider an open interval $I_0 \subset \R$ such that the fractional moment bound \eqref{eq:fracmom} holds with exponent $\mu_s > 2 \ln K$. We fix $E \in I_0$ and assume that $n \in \N$ is large enough such that $I_n \subset I_0$. 
 
We choose $a, b \in \R$ such that $I = (a,b)$ and $I_n =  (E+a/n, E+b/n)$. Let now $E_j$, $j = 1, \dots, \nu_n^{(E)}(I)$, denote the eigenvalues of $H_n$ in $I_n$. As before, for an eigenvalue $E_j \in \sigma(H_n)$ we denote by $\phi_j$ the corresponding eigenvector. 

Let us write $\epsilon_n = 6 \sqrt K  \tau_n$ for short and consider the set  $I_n^{\e}$, a small set around the endpoints of the interval $I_n$:
\[
I_n^{\e} = (E+a/n-\epsilon_n,E+a/n+\epsilon_n) \cup  (E+b/n-\epsilon_n,E+b/n+\epsilon_n) \, .
\]
By $N(I_n^{\e})$  we denote the number of eigenvalues of $H_n$ in $I_n^{\e}$. The Wegner estimate \eqref{wegner2} implies
\begin{equation}
\label{shortwegner}
\p \left[ N(I_n^{\e}) \geq 1 \right] \leq \| \rho \|_\infty |I_n^{\e}| n \leq 4  \| \rho \|_\infty  \epsilon_n n \, .
\end{equation}
To apply Lemma~\ref{lem:approx} we also need the following estimate. By $\Delta(\epsilon_n)$ we denote the event that there are two distinct eigenvalues $E_j, E_k \in \sigma(H_n) \cap I_n$ satisfying $|E_j - E_k| \leq 2\epsilon_n$. By Lemma \ref{lem:double} we have
\begin{equation}
\label{eq:double}
\p \left[ \Delta(\epsilon_n) \right] \leq 4 \| \rho \|^2_\infty |I_n| \epsilon_n n^2 = 4 \|\rho\|^2_\infty |I| \epsilon_n n \, .
\end{equation}
Finally, we remark that the event $\{ |\psi^{(x)}(x)| = |\psi^{(y)}(y)| \}$ has probability zero for any two distinct vertices $x,y \in \mathcal F_n$. Thus the event
\[
O_n = \bigcup_{\begin{subarray}{c} (x,y) \in \mathcal F_n \times \mathcal F_n \\ x \neq y \end{subarray}} \left\{ |\psi^{(x)}(x)| = |\psi^{(y)}(y)| \right\}
\]
has probability zero. By $\Delta(\epsilon_n)^c$ and $O_n^c$ we denote the complement of the corresponding event.

\begin{lemma}
\label{lem:upper}
On the event $\Omega_n = \{ N(I_n^{\e}) = 0 \} \cap \Delta(\epsilon_n)^c \cap O_n^c$ we have $\eta_n^{(E)}(I) \leq \nu_n^{(E)}(I)$.
\end{lemma}

\begin{proof}
Recall the definition of $\eta_n$ as the number of vertices in $\mathcal E_n$, see \eqref{eq:auxdef}. We write $\mathcal E_n = \{ e_i \}_{i = 1}^{\eta_n}$. By definition, each vertex $e_i$ lies in $\mathcal F_n$. Thus there is $\xi^{(e_i)} \in \sigma(H_n^{(e_i)}) \cap I_n$ and $E_{j_i} \in \sigma(H_n)$ such that \eqref{error} holds for $E_{j_i}$ and $\xi^{(e_i)}$. On the event $\Omega_n$ we have $N(I_n^{\e}) = 0$ and it follows that $E_{j_i} \in I_n$. Hence, with each vertex $e_i$, $i = 1, \dots, \eta_n$, we can associate an eigenvalue $E_{j_i} \in \sigma(H_n) \cap I_n$.

Now assume that one eigenvalue $E_j \in \sigma(H_n) \cap I_n$ is associated with two vertices $e_i, e_k \in \mathcal E_n$ such that $j_i = j_k$ and   $E_{j_i} = E_{j_k}$. Then, by \eqref{error}, we have
\begin{equation}
\label{eigenvalueclose}
\left| \xi^{(e_i)} - \xi^{(e_k)} \right| \leq \left| \xi^{(e_i)} - E_{j_i} \right|  + \left| E_{j_k} - \xi^{(e_k)} \right|  \leq 2 \sqrt K \tau_n \, .
\end{equation}
Recall from Section~\ref{ssec:construction} that the existence of $E_{j_i}$ and $E_{j_k}$ is implied by Lemma~\ref{lem:approx} with $H = H_n$ and $\lambda = \xi^{(e_i)}$ and $\lambda = \xi^{(e_k)}$ respectively.
On $\Omega_n$ the event $\Delta(\epsilon_n)$ does not happen so that the eigenvalues of $H_n$ in $I_n$ have distance at least $2 \epsilon_n$ from each other. We also have $N(I_n^\e) = 0$ and thus all conditions of Lemma \ref{lem:approx} are satisfied. Hence, the eigenvectors $\psi^{(e_i)}$ and $\psi^{(e_k)}$ corresponding to $\xi^{(e_i)}$ and $\xi^{(e_k)}$ respectively, satisfy  
\[
\left| \left< \psi^{(e_i)}, \phi_j \right>  \right|^2 \geq 1- K \tau_n^2 \epsilon_n^{-2} \quad \textnormal{and} \quad \left| \left< \psi^{(e_k)}, \phi_j \right> \right|^2 \geq 1- K \tau_n^2 \epsilon_n^{-2} \, ,
\]
where $\phi_j$ denotes the eigenvector of $H_n$ corresponding to $E_{j_i} = E_{j_k}$. Thus Lemma \ref{lem:overlap} and the identity $\epsilon_n = 6 \sqrt K \tau_n$  imply
\[
\left| \left< \psi^{(e_i)}, \psi^{(e_k)} \right> \right| \geq 1- 2 K \tau_n^2 \epsilon_n^{-2} > 0
\]
and
\begin{equation}
\label{eq:centerclose}
d(e_i,e_k) \leq 2 R_n \, .
\end{equation}
The relations \eqref{eigenvalueclose} and \eqref{eq:centerclose} show that the vertex $e_i$ belongs to the cluster $\mathcal C(e_k)$. Both vertices $e_i$ and $e_k$ lie in $\mathcal E_n$, so it follows that $|\psi^{(e_i)}(e_i)| = |\psi^{(e_k)}(e_k)|$ and thus, on $O_n^c$, $e_i = e_k$. Hence, for each vertex  $e_i$, $i = 1, \dots, \eta_n$, there is a distinct eigenvalue $E_{j_i} \in \sigma(H_n) \cap I_n$ and the proof is complete. 
\end{proof}

In order to prove the lower bound $\eta_n^{(E)}(I) \geq \nu_n^{(E)}(I)$ we introduce the following local events. 
Similar as above, we denote by $\Delta^{(x)}(\epsilon_n)$ the event that there are two distinct eigenvalues
\[
\xi_j^{(x)}, \xi_k^{(x)} \in \sigma \lk H_n^{(x)} \rk \cap I_n
\]
satisfying $|\xi_j^{(x)}-\xi_k^{(x)}| \leq 2 \epsilon_n$. Then Lemma \ref{lem:double} yields, for all $x \in \G_n$,
\begin{equation}
\label{eq:localdouble}
\p \left[ \Delta^{(x)}(\epsilon_n) \right] \leq 4 \| \rho \|_\infty^2 |I_n| \epsilon_n |B_n(x)|^2 \leq 36 \| \rho \|_\infty^2 |I| K^{2R_n } \epsilon_n n^{-1} \, .
\end{equation}

We also need the following estimate, a variant of Lemma~\ref{lem:approx}. In the remainder of this section we choose $\ln K < \mu < \mu_s - \ln K$ as in Proposition~\ref{pro:locas} and write $\mu_K = \mu - (\ln K)/2$ for short. Moreover, for $\phi \in \ell^2(\G_n)$ and $\hat x \in \G_n$ we write $\left. \phi \right|_{B_n(\hat x)}$ for the restricted vector that is set to be zero outside of $B_n(\hat x)$.

\begin{lemma}
\label{lem:revapprox}
Let $\phi$ be a $\ell^2(\G_n)$-normalized eigenvector of $H_n$ with corresponding eigenvalue $E$. Assume that there is a constant 
\begin{equation}
\label{vectorconst}
0 < C_n \leq \sqrt{ \frac {e^{2 \mu_K}-1}3 }  e^{\mu_K R_n}
\end{equation}
and  a vertex $\hat x \in \G_n$ such that, for all $x \in \G_n$,
\begin{equation}
\label{revbound}
|\phi(x)| \leq C_n e^{-\mu d(x,\hat x)} \, .
\end{equation}
Then there exists an eigenvalue $\xi \in \sigma ( H_n^{(\hat x)} )$ satisfying
\[
|E - \xi | \leq \sqrt{3 K} C_n   e^{- \mu_K (R_n+1)} 
\]
and the restricted vector $\left. \phi \right|_{B_n(\hat x)}$ satisfies
\begin{equation}
\label{renormalized}
\| \left. \phi \right|_{B_n(\hat x)} \|^2 \geq 1 -  \frac 32 C_n^2 \frac{e^{-2 \mu_K R_n}}{e^{2\mu_k}-1} \, .
\end{equation}

Assume, in addition, that the event $\Delta^{(\hat x)}(\epsilon_n)$ does not happen. Then the $\ell^2(B_n(\hat x))$-normalized eigenvector $\psi$ of $H_n^{(\hat x)}$ corresponding to $\xi$ satisfies
\[
\left| \left<  \psi, \tilde \phi \right> \right|^2 \geq 1 - 3 K C_n^2 e^{-2\mu_K (R_n+1)} \epsilon_n^{-2} \, ,
\]
where $\tilde \phi$ denotes the normalized vector $\tilde \phi = \| \left. \phi \right|_{B_n(\hat x)} \|^{-1} \left. \phi \right|_{B_n(\hat x)}$.
\end{lemma}

\begin{proof}
We have to show that $\tilde \phi$ is an approximate eigenvector for the operator $H_n^{(\hat x)}$ with domain $\ell^2(B_n(\hat x))$. The operators $H_n$ and $H_n^{(\hat x)}$ coincide in the interior of $B_n(\hat x)$ such that 
\[
\left\| \lk H_n^{(\hat x)}  - E \rk \phi \right\|_{\ell^2(B_n(\hat x))}^2 = \sum_{x \in \partial B_n(\hat x)} \left| \sum_{y \notin B_n(\hat x) : d(y,x) = 1} \phi(y) \right|^2 \, .
\]
Hence, from the Schwarz inequality and from  \eqref{revbound} and \eqref{surfbound} we obtain
\begin{equation}
\label{notnormalizedbound}
\left\| \lk H_n^{(\hat x)}  - E \rk \phi \right\|_{\ell^2(B_n(\hat x))}^2 \leq K \sum_{\partial B_{R_n+1}(\hat x)} \left| \phi(y) \right|^2  \leq \frac 32 K C_n^2 e^{-2 \mu_K (R_n+1)} \, .
\end{equation}

To estimate the effect of the normalization of $\tilde \phi$ we use \eqref{revbound} and \eqref{surfbound} again to get 
\begin{align*}
\sum_{y \notin B_n(\hat x)} |\phi(y)|^2  &= \sum_{m > R_n} \sum_{d(y, \hat x) = m} | \phi(y)|^2 \leq \frac 32 \sum_{m > R_n}   K^{m} C_n^2 e^{-2\mu m} = \frac 32 C_n^2 \frac{e^{-2 \mu_K R_n}}{e^{2\mu_k}-1} \, .
\end{align*}
Since $\phi$ is normalized this implies \eqref{renormalized}.
Inserting this into  \eqref{notnormalizedbound} and applying \eqref{vectorconst} to simplify yields
\[
\left\| \lk H_n^{(\hat x)}  - E \rk \tilde \phi \right\|_{\ell^2(B_n(\hat x))}^2 \leq \frac{3 K C_n^2 e^{-2 \mu_K (R_n+1)}}{2 - 3C_n^2 (e^{2\mu_K}-1)^{-1} e^{-2\mu_K R_n}} \leq 3K C_n^2 e^{-2 \mu_K (R_n+1)} \, .
\]
The remainder of the proof follows the same arguments as the proof of Lemma \ref{lem:approx}.
\end{proof}

Now we proceed to give the proof of the lower bound.

\begin{lemma}
\label{lem:goodset}
Let the constant $C_n > 0$ satisfy 
\begin{equation}
\label{secondvectorconst}
C_n \leq \tau_n^2  e^{(\mu - \ln K)R_n} \, .
\end{equation}
Define the event
\[
\Omega_n' = \{ N(I_n^{\e}) = 0 \} \cap \Delta(\epsilon_n)^c \cap \left\{ X_n \leq C_n \right\} \cap \bigcap_{x \in \G_n}    \Delta^{(x)}(\epsilon_n)^c \cap O_n^c 
\]
with $X_n$ from Proposition~\ref{pro:locas}.
Then, for $n$ large enough, on $\Omega_n'$ we have 
\[
\eta_n^{(E)}(I) =  \nu_n^{(E)}(I) \, .
\]
\end{lemma}

\begin{proof}
The upper bound follows directly from Lemma~\ref{lem:upper} and the fact that $\Omega_n' \subset \Omega_n$. Let us proceed to prove the lower bound $\eta_n \geq \nu_n$.

Here we use the assumption that within $I_0$ the fractional moment bound \eqref{eq:fracmom} is satisfied: For $n$ large enough we have $I_n \subset I_0$ and Proposition~\ref{pro:locas} yields that for each eigenvalue $E_j \in \sigma(H_n) \cap I_n$, the corresponding eigenvector $\phi_j$ satisfies \eqref{eq:loccond}. In addition, on $\Omega'_n$ we have $X_n(\omega) \leq C_n$. Thus there exists a vertex $x_j \in \G_n$ such that, for all $x \in \G_n$, 
\begin{equation}
\label{eq:evgoodbound}
|\phi_j(x)| \leq C_n e^{- \mu d(x,x_j)} \, .
\end{equation}
First, we need to show that $x_j \in \mathcal F_n$, so  we have to verify conditions (i) and (ii) from Section~\ref{ssec:construction}.

Note that \eqref{secondvectorconst} combined with $\tau_n \leq 1/6\sqrt K n$ implies \eqref{vectorconst}. Hence, Lemma \ref{lem:revapprox} applied to $\phi_j$ yields existence of an eigenvalue $\xi^{(x_j)} \in \sigma ( H_n^{(x_j)})$ satisfying $\left| E_j - \xi^{(x_j)} \right| \leq \sqrt{3 K} C_n e^{-\mu_K ( R_n+1)}$. From \eqref{secondvectorconst} and the fact that $\tau_n \leq 1/6\sqrt K n$ we conclude 
\[
\left| E_j - \xi^{(x_j)} \right| \leq \sqrt K \tau_n \, .
\]
On $\Omega'_n$, we have $N(I_n^{\e})=0$ and it follows that $\xi^{(x_j)} \in I_n$ and that $x_j$ satisfies (i).

To verify condition (ii) let $\psi^{(x_j)}$ be the eigenvector corresponding to $\xi^{(x_j)}$ and let $\tilde \phi_j$ denote the normalized restricted eigenvector corresponding to $E_j$: 
\[
\tilde \phi_j = \left\| \left. \phi_j \right|_{B_n(x_j)}\right\|^{-1} \left. \phi_j \right|_{B_n(x_j)} \, .
\]
The event $\Delta^{(x_j)}(\epsilon_n)$ does not occur on $\Omega_n'$, thus Lemma~\ref{lem:revapprox} implies 
\[
\left| \left<  \tilde \phi_j, \psi^{(x_j)} \right> \right|^2 \geq 1 - 3 K C_n^2   e^{-2\mu_K (R_n+1)} \epsilon_n^{-2} \, .
\]
Hence, by Lemma \ref{lem:overlapbound} and \eqref{surfbound}, we find
\[
\sum_{y \in \partial B_n(x_j)} \left | \psi^{(x_j)} (y) \right|^2 \leq 2 \lk \sum_{y \in \partial B_n(x_j)} \left | \tilde \phi_j (y) \right|^2 + \frac 92 K^{R_n+1}  C_n^2  e^{-2 \mu_K ( R_n+1)} \epsilon_n^{-2} \rk \, .
\]
By  \eqref{surfbound}, \eqref{renormalized}, \eqref{vectorconst},  and \eqref{eq:evgoodbound}, the first summand is bounded by $3 C_n^2 e^{-2 \mu_K R_n}$. Inserting the identities $\epsilon_n = 6 \sqrt K \tau_n$ and $\mu_K = \mu - (\ln K)/2$ and using the fact that $\tau_n \leq 1/6 \sqrt K n$ we obtain
\begin{align*}
\sum_{y \in \partial B_n(x_j)} \left | \psi^{(x_j)} (y) \right|^2 &\leq \frac 14 C_n^2 K^{2R_n+1}  e^{-2\mu (R_n+1)} \tau_n^{-2} \lk \frac{24 e^{2 \mu}\tau_n^2}{K^{R_n+1}} +1 \rk\\
&\leq  C_n^2 e^{-2(\mu-\ln K)R_n} \tau_n^{-2} \, .
\end{align*}
By \eqref{secondvectorconst} we see that the right-hand side is bounded by $\tau_n^2$ so that $x_j$ satisfies condition (ii). Hence, $x_j \in \mathcal F_n$.

To associate with $E_j$ a vertex from $\mathcal E_n$ write $\mathcal E_n = \{ e_i \}_{i=1}^{\eta_n}$ and consider
\[
e_{i_j} = \argmax_{y \in \mathcal C(x_j)} \left| \psi^{(y)}(y) \right| \, .
\]
We will now show that 
\begin{equation}
\label{eine}
e_{i_j} \in \mathcal E_n \, .
\end{equation}
Let $u \in \mathcal C(e_{i_j})$. Then, since $e_{i_j} \in \mathcal C(x_j)$,
\[
\left| \xi^{(u)} - \xi^{(x_j)} \right| \leq \left|  \xi^{(u)} - \xi^{(e_{i_j})} \right|  + \left| \xi^{(e_{i_j})} - \xi^{(x_j)} \right| \leq 4 \sqrt K \tau_n \, .
\]
By \eqref{error}, there is an eigenvalue $E^{(u)} \in \sigma(H_n)$ such that $\left| E^{(u)} - \xi^{(u)} \right| \leq  \sqrt K \tau_n$. Combining this with $6 \sqrt K \tau_n = \epsilon_n$ we get
\[
\left| E^{(u)} - E_j \right| \leq \left| E^{(u)} - \xi^{(u)} \right| +\left| \xi^{(u)} - \xi^{(x_j)} \right| + \left| \xi^{(x_j)} - E_j \right|  \leq 6 \sqrt K \tau_n = \epsilon_n \, .
\] 
On $\Omega_n'$ the event $\Delta(\epsilon_n)$ does not happen, so we conclude that $E^{(u)} = E_j$. In particular, we obtain
\[
\left| \xi^{(u)} - \xi^{(x_j)} \right| \leq \left| \xi^{(u)} - E^{(u)} \right| + \left| E_j  - \xi^{(x_j)} \right| \leq 2 \sqrt K \tau_n \, .
\]
Moreover, we can apply Lemma~\ref{lem:approx} using the same argument that lead to \eqref{eq:centerclose} to show $d(u,x_j) \leq 2 R_n$ and we find $u \in \mathcal C(x_j)$. This shows that $\mathcal C(e_{i_j}) \subset \mathcal C(x_j)$ so that
\[
e_{i_j} = \argmax_{y \in \mathcal C(e_{i_j})} \left| \psi^{(y)}(y) \right| \, .
\]
On $O_n^c$, this proves \eqref{eine}.
So with each eigenvalue $E_j \in \sigma(H_n) \cap I_n$, $j = 1, \dots, \nu_n$, we can associate a vertex $e_{i_j} \in \mathcal E_n$.

Assume now that $e_{i_j} = e_{j_k}$ for $j,k \in \{1, \dots, \nu_n\}$. In particular, we get $\xi^{(e_{i_j})} = \xi^{(e_{i_k})}$ and, since $e_{i_j} \in \mathcal C(x_j)$ and $e_{i_k} \in \mathcal C(x_k)$,  
\[
\left| E_j - E_k \right| \leq \left| E_j - \xi^{(x_j)} \right| + \left| \xi^{(x_j)} - \xi^{(e_{i_j})}  \right| +  \left| \xi^{(x_k)} - \xi^{(e_{i_k})}  \right| +  \left| \xi^{(x_k)} - E_k \right|  \leq 6 \sqrt K \tau_n = \epsilon_n \, .
\]
On $\Omega_n'$, it follows that $E_j = E_k$. So with each eigenvalue $E_j$, $j = 1, \dots, \nu_n$, we can associate a distinct vertex $e_{i_j}$. This proves the lower bound $\eta_n \geq \nu_n$ and completes the proof.
\end{proof}

With Lemma~\ref{lem:goodset} at hand we can complete the proof of Proposition~\ref{pro:compare}.

\begin{proof}[Proof of Proposition \ref{pro:compare}]
We note the obvious bounds $\nu_n^{(E)}(I) \leq n$ and $\eta_n^{(E)}(I) \leq n$. Thus Lemma~\ref{lem:goodset} shows
\[
\E \left[ \left| \nu_n^{(E)}(I)  - \eta_n^{(E)}(I) \right| \right] \leq n \p \left[{ \Omega'_n}^c \right] 
\]  
and
\[
\E \left[ \left| e^{-t \nu_n^{(E)}(I)} - e^{-t  \eta_n^{(E)}(I) } \right| \right] \leq  \p \left[{ \Omega'_n}^c \right] 
\]  
for all $t \geq 0$.
Combining \eqref{shortwegner}, \eqref{eq:double}, \eqref{eq:exploc}, and \eqref{eq:localdouble} yields
\[
\p \left[{ \Omega'_n}^c \right] \leq 4 \| \rho \|_\infty \epsilon_n n + 4 \| \rho \|_\infty^2 |I| \epsilon_n n + \Cl n |I_n| C_n^{-2(1-a)/a} + 36 \| \rho \|_\infty^2 |I| K^{2R_n} \epsilon_n
\]
and the claim follows from the identities $|I_n| = |I|/n$ and $\epsilon_n = 6 \sqrt K \tau_n$.
\end{proof}


\appendix

\section{Exponential localization of eigenvectors}
\label{ap:exp}

In this section we show that exponentially decaying bounds on fractional moments of the Green function imply  exponential localization of eigenvectors as stated in Proposition~\ref{pro:locas}.

\begin{proof}[Proof of Proposition \ref{pro:locas}]
This result follows from an estimate of the eigenfunction correlator in the spirit of \cite{Aiz94,AizSchFriHun01,AizElgNabSchSto06}. Let $P_{ \{ \cdot \} }$ denote the spectral projection corresponding to $H_n(\omega)$ and for $x, y \in \G_n$ define the eigenfunction correlator
\[
Q(x,y;I) = \sum_{E \in \sigma(H_n) \cap I} \left| \left< \delta_x, P_{\{E\}} \delta_y \right> \right| = \sum_{E \in \sigma(H_n) \cap I} |\phi_E(x)| \, |\phi_E(y)| \, ,
\]
where $\phi_E$ denote the eigenvector corresponding to $E \in \sigma(H_n(\omega))$. It was noticed in \cite{Aiz94} that the eigenfunction correlator can be bounded in terms of fractional moments of the Green function. Using the Schwarz inequality we estimate
\begin{align*}
\sum_{E \in \sigma(H_n) \cap I} |\phi_E(x)| \, |\phi_E(y)| \leq & \lk \sum_{E \in \sigma(H_n) \cap I} |\phi_E(x)|^{2-s} \, |\phi_E(y)|^s \rk^{1/2} \\
& \times \lk \sum_{E \in \sigma(H_n) \cap I} |\phi_E(x)|^s \, |\phi_E(y)|^{2-s} \rk^{1/2} \, .
\end{align*} 
Let $\sigma_x$ denote the spectral measure of $H_n(\omega)$ associated with vertex $x \in \G_n$. Then it follows that
\begin{align*}
Q(x,y;I) \leq & \lk \int_I  \left| \left< \delta_x, P_{\{E\}} \delta_x \right> \right|^{-s} \left| \left< \delta_x, P_{\{E\}} \delta_y \right> \right|^s \sigma_x(dE) \rk^{1/2} \\
& \times \lk \int_I  \left| \left< \delta_y, P_{\{E\}} \delta_y \right> \right|^{-s} \left| \left< \delta_y, P_{\{E\}} \delta_x \right> \right|^s \sigma_y(dE) \rk^{1/2} \, .
\end{align*}
Now one can apply the results of \cite[Section 3]{Aiz94} to estimate  
\[
\E\left[ Q(x,y;I) \right] \leq C \int_{I} \E \left[ \left| G_{n,\alpha}(x,y;E+i0) \right|^s \right] dE \, ,
\]
see also \cite[Thm. 6.8]{AizWar14} for a streamlined presentation. 
Hence, the fractional moment bound \eqref{eq:fracmom} implies
\begin{equation}
\label{eq:correlator}
\E\left[ Q(x,y;I) \right] \leq C |I| e^{-\mu_s d(x,y)} 
\end{equation}
with a uniform constant $C > 0$.

Let us now choose $\mu$ such that $\ln K < \mu < \mu_s - \ln K$ and estimate
\begin{equation}
\label{eq:corr2}
Q(x,y;I) = Q(x,y;I) e^{\mu d(x,y)} e^{-\mu d(x,y)} \leq A(x,\omega) e^{- \mu d(x,y)} \, ,
\end{equation}
where we have introduced the random variable
\[
A(x,\omega) =  \sum_{v \in \G_n} e^{\mu d(x,v)} Q(x,v;I) \, .
\]
By \eqref{eq:correlator} and \eqref{surfbound}, we have, for each fixed $x \in \G_n$,
\[
\E \left[ A(x) \right] \leq C |I| \sum_{v \in \G_n} e^{-(\mu_s-\mu) d(x,v)} \leq  \frac 32 \frac{C |I|}{1-K e^{-(\mu_s-\mu)}} \, .
\]
Now we set,  for  $1 < a < 2 - \ln K/\mu$,
\[
X_n(\omega) =\lk \frac 32 \frac{1}{1-K e^{-(2-a)\mu}} \rk^{1/(2a-2)} \lk  \sum_{x \in \G_n}  A(x,w) \rk^{a/(2a-2)} 
\]
and we obtain 
\[
\E \left[ X_n^{2(a-1)/a} \right] \leq \frac 32 \frac{C |I| n}{1-K e^{-(\mu_s-\mu)}}  \lk \frac 32 \frac{1}{1-K e^{-(2-a)\mu}} \rk^{1/a} \, .
\]
This yields \eqref{eq:exploc} with 
\[
\Cl = \frac 32 \frac{C }{1-K e^{-(\mu_s-\mu)}}  \lk \frac 32 \frac{1}{1-K e^{-(2-a)\mu}} \rk^{1/a} \, .
\]

To show that \eqref{eq:loccond} follows, we choose,  for each eigenvector $\phi_j$ with eigenvalue $E_j \in I$ the random vertex $x_j(\omega) = \argmax |\phi_j(x)|$. Then the definition of $Q(x,y;I)$ combined with estimate \eqref{eq:corr2} implies, for all $x \in \G_n$,
\begin{equation}
\label{eq:ef:apriori}
|\phi_j(x)| \leq \frac 1{|\phi_j(x_j)|} Q(x,x_j;I) \leq \frac {A(x_j,\omega)}{|\phi_j(x_j)|} e^{-\mu d(x,x_j)} \, .
\end{equation}
To estimate $|\phi_j(x_j)|$ from below we note that the normalization of $\phi_j$ implies, for any $1 < a < 2 - \ln K/\mu$, 
\[
1 = \sum_{x \in \G_n} |\phi_j(x)|^2 = \sum_{x \in \G_n} |\phi_j(x)|^a |\phi_j(x)|^{2-a}   \leq  |\phi_j(x_j)|^a \sum_{x \in \G_n}  |\phi_j(x)|^{2-a} \, .
\]
Inserting \eqref{eq:ef:apriori} and then using \eqref{surfbound} yields
\begin{align*}
1 &\leq |\phi_j(x_j)|^{2a-2} A(x_j,\omega)^{2-a} \sum_{x \in \G_n} e^{-(2-a)\mu d(x,x_j)} \\
& \leq |\phi_j(x_j)|^{2a-2} A(x_j,\omega)^{2-a} \frac 32 \frac{1}{1-Ke^{-(2-a)\mu}} \, .
\end{align*}
Solving for $|\phi_j(x_j)|$ we obtain
\[
|\phi_j(x_j)| \geq \lk \frac 23 \lk 1-K e^{-(2-a)\mu} \rk \rk^{1/(2a-2)} \frac 1{A(x_j,w)^{(2-a)/(2a-2)}} \, .
\]
Finally, we insert this estimate into \eqref{eq:ef:apriori} and we see that 
\[
|\phi_j(x)| \leq \lk \frac 32 \frac{1}{1-K e^{-(2-a)\mu}} \rk^{1/(2a-2)} A(x_j,w)^{a/(2a-2)} e^{-\mu d(x,x_j)} \leq X_n(\omega) e^{-\mu d(x,x_j)} \, .
\]
This completes the proof.
\end{proof}


\subsection*{Acknowledgments} The author wants to thank Michael Aizenman and Simone Warzel for helpful discussions. Financial support from DFG grant GE 2369/1-1 and NSF grant PHY-1122309 is gratefully acknowledged.



\end{document}